\documentclass{article}
\usepackage{graphicx} 
\usepackage{float}
\usepackage{float}
\usepackage{amsmath}  
\usepackage{amssymb}  
\usepackage{authblk}
\usepackage{tikz}
\usetikzlibrary{shapes.geometric, arrows.meta, positioning, calc}
\usepackage{amsthm}
\newtheorem{definition}{Definition}[section]

\newtheorem{theorem}[definition]{Theorem}
\newtheorem{proposition}[definition]{Proposition}
\newtheorem{corollary}[definition]{Corollary}
\usepackage{enumerate}
\usepackage{url}

\title{A Disproof of Large Language Model Consciousness: The Necessity of Continual Learning for Consciousness}
\author[1]{Erik Hoel\thanks{\texttt{erik@bicameral-labs.org}}}
\affil[1]{Bicameral Labs, MA, USA}

\begin{document}

\maketitle

\section{Abstract}
Scientific theories of consciousness should be falsifiable and non-trivial. Recent research has given us formal tools to analyze these requirements of falsifiability and non-triviality for theories of consciousness. Surprisingly, many contemporary theories of consciousness fail to pass this bar, including theories based on causal structure but also (as I demonstrate) theories based on function. Herein, I show these requirements of falsifiability and non-triviality especially constrain the potential consciousness of contemporary Large Language Models (LLMs) because of their proximity to systems that are equivalent to LLMs in terms of input/output function; yet, for these functionally equivalent systems, there cannot be \textit{any} falsifiable and non-trivial theory of consciousness that judges them conscious. This forms the basis of a disproof of contemporary LLM consciousness. I then show a positive result, which is that theories of consciousness based on (or requiring) continual learning do satisfy the stringent formal constraints for a theory of consciousness in humans. Intriguingly, this work supports a hypothesis: If continual learning is linked to consciousness in humans, the current limitations of LLMs (which do not continually learn) are intimately tied to their lack of consciousness.

\section{Constraining theories of consciousness to testable theories}

In the past decade, a subfield of consciousness research has taken a ``structural turn'' and become focused on mathematics and formalisms \cite{kleiner2020mathematical,kleiner2024towards}. Much of this is the influence of Integrated Information Theory (IIT) \cite{albantakis2023integrated}, since it is expressed in a mathematically formalized way \cite{barbosa2020measure,kleiner2021mathematical}. In turn, many of the criticisms of IIT have depended on its detailed particulars and whether it is scientifically testable \cite{bayne2018axiomatic,fleming2023integrated}. These debates have led to increased interest in defining minimal ``toy models'' of consciousness \cite{hoel2024world,albantakis2025utility} and identifying what requirements might constrain theories of consciousness in terms of being testable (such as being falsifiable).

This structural turn gives rise to a novel approach: \textit{Can a theory of consciousness be so strongly constrained by requirements for testability that its overall nature can be deduced?}

Increasingly, such an approach seems a necessity. Currently, the field of consciousness research has hundreds of existing theories \cite{seth2022theories,kuhn2024landscape}, and no agreed-upon way to empirically distinguish them. The search for the neural correlates of consciousness \cite{crick1990towards,koch2016neural} has revealed that even tightly-controlled and well-funded adversarial collaborations fail \cite{cogitate2025adversarial}, such as the recent head-to-head comparison of IIT and Global Neuronal Workspace Theory \cite{mashour2020conscious, baars1993cognitive}, which led to accusations of pseudoscience \cite{fleming2023integrated, tononi2025consciousness}. If consciousness research is to make serious progress, it must winnow the wide field via constraints on properties of theories \cite{kleiner2021falsification, doerig2021hard}. In this, scientists would act like artists drawing the negative space around consciousness to see its outline clearly.

A further reason to pursue a testability-first approach is that we want questions about consciousness that the contemporary state of consciousness research cannot answer yet. Most prominently, the question whether or not contemporary LLMs (like ChatGPT or Claude or Gemini) are conscious has become suddenly critical. There are major risks associated with getting this question wrong. Assigning consciousness where there is none has a myriad of risks, which include increasing the risk of AI psychosis, overestimation of LLM capabilities, inappropriate practices or regulation, and misleading scientific beliefs about human consciousness \cite{seth2024conscious, morrin2025delusions}. On the other hand, if contemporary LLMs were conscious they could be considered moral patients \cite{sebo2025moral} and be due considerations like conversational ``exit rights'' \cite{anthropic2025exit} as well as other protections against mistreatment.

It is unlikely that empirical data from LLMs alone can answer the question of their consciousness, given how prone they are to confabulation and prompt-sensitivity \cite{ji2023survey}. And while there is a small subfield studying LLM ``introspection'' \cite{binder2024looking}, the best evidence for this introspection is in the form of statistically uncommon anomaly detection of injected concepts, for which there could be shallow mechanistic explanations \cite{anthropic2025introspection}. 

Therefore, most have assumed that answering the question of LLM consciousness would require scientific consensus around a particular theory of consciousness. However, I show that this is not necessary. Instead, it is possible to prove that there is \textit{no} theory of consciousness that is non-trivial and falsifiable that \textit{could} exist for (at minimum) baseline LLMs. This represents much faster progress than trying to assign assumption-dependent probabilities of consciousness, or trying to apply some set of the many (often contradictory) theories that do currently exist \cite{butlin2023consciousness}.

This road to a disproof of LLM consciousness begins with one particularly influential criticism of IIT: The ``Unfolding Argument'' by Doerig et al. \cite{doerig2019unfolding} has led to debate in the literature on how to falsify theories of consciousness. Specifically, the argument is based on the universal approximation theorem \cite{hornik1989multilayer}, and argues that any given Recurrent Neural Network (RNN) can be ``unfolded'' into some functionally equivalent Feedforward Neural Network (FNN). This creates a pathological scientific situation for theories like IIT, since any given RNN might have a certain degree of integrated information, yet its ``twin'' FNN would have zero integrated information. This is despite the two networks sharing the same input/output (I/O).

Later work by Johannes Kleiner and myself \cite{kleiner2021falsification} showed that unfolding was a specific case of the independence between what a theory uses to make predictions about consciousness (like the integrated information of IIT) and how a scientist infers the actual states of conscious experiences during experiments (via things like report or behavior or, more broadly, any I/O). We proposed a framework wherein alternative systems (of any kind) can be ``substituted'' in during empirical testing while holding I/O fixed, which reveals pathologies around falsifiability (thus, the "Substitution Argument").

Additionally, we pointed out that, instead of independence, there is a risk of having \textit{too much} dependency between a theory's predictions and a scientist's inferences. This would create theories of consciousness that are trivial. In other words, there are two ``horns'' for theories of consciousness: the Substitution Argument on one side, and triviality on the other. The dual-horned ``Kleiner-Hoel dilemma'' established in \cite{kleiner2021falsification} creates only a small space through which to pass; theories of consciousness are ``caught between Scylla and Charybdis.''

According to the Substitution Argument, theories like IIT \cite{tononi2012integrated} or recurrent processing theories \cite{lamme2006towards} are a priori falsified by specific substitutions like unfolded neural networks \cite{doerig2019unfolding}. Meanwhile, theories like behaviorism or functionalism based solely on I/O are trivial, because predictions of any such theory can never differ from inferences about consciousness from experiments (which are also based on I/O). E.g., if a theory of consciousness is that consciousness is identical to behaving as if conscious, how could it ever be falsified? It gives only trivial scientific information.

This paper shows that the Kleiner-Hoel dilemma is particularly troublesome for LLMs, and gives a disproof of their consciousness.

First, I overview a formal framework for testing theories of consciousness \cite{kleiner2021falsification}, and demonstrate how the Kleiner-Hoel dilemma naturally arises from it, including a new proof based on Kolmogorov's notion of a ``shortest program'' for how even functionalist theories face challenges from substitutions (or, depending on how they're defined, fall into triviality).

Second, I show that the Kleiner-Hoel dilemma has especially strong force in LLMs, because LLMs are much closer than human brains in ``substitution space'' to I/O equivalent substitutions like lookup tables that non-trivial theories of consciousness can't possibly judge as conscious. That is, there is no theory of consciousness available for contemporary LLMs that would not fall on one of the horns of the Kleiner-Hoel dilemma. Therefore, there is no scientific theory of consciousness which could judge contemporary LLMs as conscious. Therefore, they are not conscious. 

The strength of this argument is that it is entirely agnostic to particular theories of consciousness being true, nor assumes the ``specialness'' of biology, nor outright dismisses the possibility of AI consciousness in general, and has few philosophical commitments (other than that theories of consciousness should be falsifiable and non-trivial).

Third, this paper sketches a solution to the Kleiner-Hoel dilemma in humans by grounding a theory of consciousness not in causal structure like IIT, nor functional or dynamical features like Global Workspace Theory, but in continual learning. This fits well with the recent observation by O’Reilly-Shah et al. that a static FNN cannot act as a universal substitution for a plastic RNN \cite{o2025caveat}. However, those recent proofs only show that plastic RNNs are protected from the specifics of the Unfolding Argument---they do not fully distinguish whether the recurrence or plasticity is the cause, nor ask if this protection extends to other universal substitutions. Here I show that grounding a theory of consciousness in the broader context of continual learning explicitly avoids \textit{both} horns of the Kleiner-Hoel dilemma.

Assuming that a robust scientific theory of consciousness exists, until other cases that avoid the Kleiner-Hoel dilemma are identified, they are strongly constrained to not merely being confined to systems that can learn (like plastic RNNs), but more generally to involve continual learning itself. Such a result significantly winnows the field of consciousness research and illuminates the path toward less popular, but ultimately more scientific, theories of consciousness. It also implies that the vast majority of the search for the neural correlates of consciousness may be incommensurate, because testing static I/O behavior alone is not how one empirically tests for consciousness.

Intriguingly, this research indicates that one major reason deployed LLMs lack consciousness is their static structure, and that unlocking continual learning, already a significant target in AI research \cite{kudithipudi2022biological}, could possibly change this disproof. It also provides grounds to believe that the current cognitive limitations of LLMs in adaptation to novelty and generalization are based on their lack of consciousness, which is connected to their lack of continual learning.

\section{A Formal Framework for Falsifying Theories of Consciousness}

The line of reasoning here will make use of a previously developed framework for falsification of theories of consciousness \cite{kleiner2021falsification}. It is based on treating theories as mapping a space of possible experiences onto predictions about the system's workings. Its advantage is that it is abstract, general, and (relatively) simple and broadly follows common philosophical definitions of falsification \cite{popper2005logic,lakatos2014falsification}.

I leave much of the background logic and reasoning to previous work (see \cite{kleiner2021falsification} for those details). For our purposes what matters is that testing a theory of consciousness is based on comparing two functions: predictions and inferences. 

First, what are predictions? Predictions pick out a particular element of the space of possible experiences, $E$. The correspondence (here, for simplicity, referred to as a function) that picks out experiences is a mapping from experimental data about the system's workings, $O$ (such as neuroimaging data), to the space of experiences, $E$; this mapping is dubbed \textit{pred}, for ``prediction.''
\[
    pred: O \rightarrow E \:.
\]
Theories are defined by their unique methods of making predictions. However, predictions alone cannot falsify a theory. They need to be compared to empirical inferences about consciousness.

Usually, these inferences are based on verbal report or behavior. In a psychology lab, button presses to indicate what was flashed in a masking stimulus could be inference data; in another context, like in assessing LLM consciousness, the outputs of an artificial neural network could be inference data. Inferences are, essentially, empirical reasons to believe a system might be conscious \textit{irrespective} of any particular theory or class of theories. Perhaps the clearest inference would be an explicit claim by the system to consciousness itself, but  others might be, e.g., correctly reporting sensory inputs, responding to pain, etc.

Inferences can be represented by a second function, which maps between the data generated by observations about the behavior of a system (like its I/O) and the space of experiences $E$:

\[
inf: O \rightarrow E \:.
\]
The combination of \textit{inf} and \textit{pred} is visualized in Figure 1.

\begin{figure}[ht]
\centering
\begin{tikzpicture}[
    >=Stealth,
    every node/.style={font=\small},
    set/.style={ellipse, draw, thick, minimum width=3.8cm, minimum height=5cm, align=center},
    element/.style={circle, fill=black!25, inner sep=0pt, minimum size=8pt},
]
\node[set] (data) at (0,0) {};
\node[font=\bfseries, above=0.1cm of data.north] {Experimental Data};
\node[set, minimum width=3cm, minimum height=4cm] (exp) at (6.5,0) {};
\node[font=\bfseries, above=0.1cm of exp.north] {Experiences};
\node[element] (reports) at (1.2, 1.0) {};
\node[anchor=east, align=right, font=\footnotesize] at (1.0, 1.0) {Data from\\verbal reports};
\node[element] (neural) at (1.2, -0.2) {};
\node[anchor=east, align=right, font=\footnotesize] at (1.0, -0.2) {Data about\\system's workings};
\node[element] (E) at (6.3, 0.4) {};
\node[anchor=west, font=\itshape] at (6.55, 0.4) {$E$};
\draw[->, thick] (reports) to[bend left=15] 
    node[midway, above=2pt, font=\footnotesize] {inferences} 
    (E);
\draw[->, thick] (neural) to[bend right=10] 
    node[pos=0.4, below=2pt, font=\footnotesize] {predictions} 
    (E);

\end{tikzpicture}
\caption{A theory of consciousness is not falsified (sometimes called ``experimentally confirmed,'' although that language is often not appropriate) when inference from behavioral reports and prediction from neural data converge on the same experience from $E$.}
\label{fig:convergence}
\end{figure}
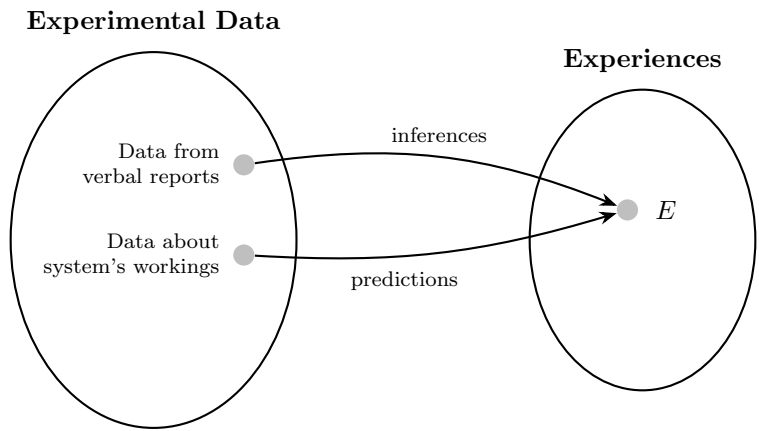

Obviously, this formal setup is an extremely idealized and abstracted version of testing a theory of consciousness (for a defense of the irrelevancy of other details, however, see \cite{kleiner2021falsification}).

It is important that a falsification framework not be overfitted or over-sensitive. For the purposes herein, I do not require that falsification triggers from just \textit{any} possible mismatch between \textit{pred} and \textit{inf}. Therefore, here I remain agnostic to whether the falsification framework introduced herein can be \textit{optionally} permissive about some occasional mismatches.

However, if a theory of consciousness predicts that extremely simple systems have degrees of consciousness greater than that of a human \cite{aaronson2014unconscious, tononi2014conscious}, such a radical mismatch could count as actual falsification---but \textit{only} in light of the inference that a system's behavioral or I/O simplicity is evidence against such a high degree of consciousness. Ideally, such inferences should always be made explicit when reasoning about consciousness through the lens of falsification. It is important to note that keeping assumptions about inferences as neutral as possible is important, or else testing a theory of consciousness would lose its desired empirical ``universality'' \cite{kanai2024toward} and become less robust.

\begin{figure}[ht]
\centering
\begin{tikzpicture}[
    >=Stealth,
    every node/.style={font=\small},
    set/.style={ellipse, draw, thick, minimum width=3.8cm, minimum height=5cm, align=center},
    element/.style={circle, fill=black!25, inner sep=0pt, minimum size=8pt},
]
\node[set] (data) at (0,0) {};
\node[font=\bfseries, above=0.1cm of data.north] {Experimental Data};
\node[set, minimum width=3cm, minimum height=4cm] (exp) at (6.5,0) {};
\node[font=\bfseries, above=0.1cm of exp.north] {Experiences};
\node[element] (reports) at (1.2, 1.0) {};
\node[anchor=east, align=right, font=\footnotesize] at (1.0, 1.0) {Data from\\verbal reports};
\node[element, fill=black!15] (neural_old) at (1.2, -0.2) {};
\node[element] (neural_new) at (1.2, -1.0) {};
\node[anchor=east, align=right, font=\footnotesize] at (1.0, -0.6) {Data about\\system's workings};
\draw[->, thick, dashed] (neural_old) -- (neural_new) node[midway, right, font=\footnotesize] {$S$};
\node[element, fill=black!15] (E_old) at (6.3, 0.4) {};
\node[anchor=west, font=\itshape, text=black!40] at (6.55, 0.4) {$E$};
\node[element] (E_new) at (6.3, -0.4) {};
\node[anchor=west, font=\itshape] at (6.55, -0.4) {$E'$};
\draw[->, thick] (reports) to[bend left=15] 
    node[midway, above=2pt, font=\footnotesize] {inferences} 
    (E_old);
\draw[->, thick, draw=black!30] (neural_old) to[bend right=10] 
    (E_old);
\draw[->, thick] (neural_new) to[bend right=15] 
    node[pos=0.4, below=2pt, font=\footnotesize] {predictions} 
    (E_new);
\draw[->, thick, dashed] (E_old) -- (E_new);

\end{tikzpicture}
\caption{The Substitution Argument: Changing the system (and thus significantly changing the prediction function), either by transformation or swapping in a different system, all while holding inference data constant, should always be possible as long as inferences and predictions are completely independent. Here, following a substitution, the predicted experience shifts from $E$ to $E'$, but the inferred experience remains $E$, creating the mismatches that falsify the theory.}
\label{fig:substitution}
\end{figure}
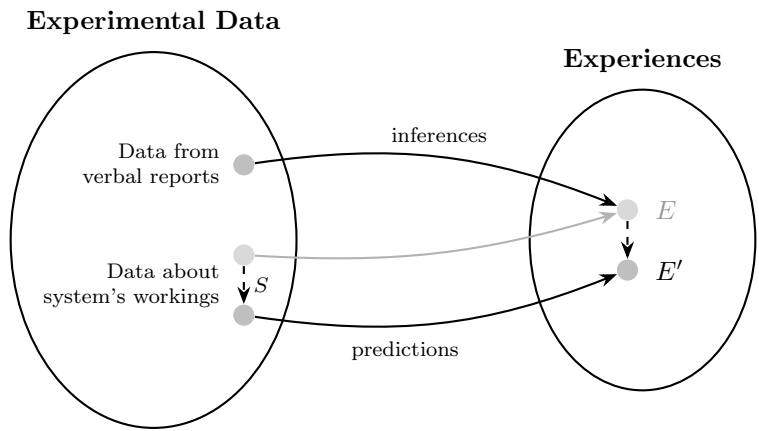

Importantly, if \textit{inf} is held constant, but \textit{pred} is varied significantly, then this will create \textit{consistent} divergence (mismatches) across the entirety (or vast majority) of a theory's predictions. This is pathological behavior and should necessarily falsify a given theory, as now \textit{every} prediction is wrong (or the vast majority are wrong) for a system with identical (or, again, extremely similar) inferences. Specifically, in \cite{kleiner2021falsification} we took the approach that a ``universal substitution'' (of one system for another) in an experimental setup (within sensible allowances), but in which inferences do not change---but all predictions do change---would completely falsify a particular theory. We called such a possibility a ``universal substitution'' and confined the analyses to such cases in order to avoid more debatable instances. E.g.,  it is unclear to what degree ``partial substitutions'' triggering mismatches should count as falsification (but again, this question can be set aside here). However, what is clear is that universal substitutions stem from predictions and inferences being totally independent.

What universal substitutions provide is a means of \textit{a priori} falsification. Essentially, substituting them into an experimental setup for the original would radically alter or complicate the predictions of a theory of consciousness, while the outward and immediate inferences (such as report or output) would not change at all (other than irrelevant considerations for testing a theory of consciousness, like say, exact timing of reports). 

There are two main classes of universal substitutions: causal-structure substitutions and algorithmic substitutions.

\begin{itemize}
    \item     Causal-structure substitutions accomplish a particular I/O function via a different system with different internal workings. This includes lookup tables that implement the same I/O function of a system directly \cite{searle1980minds}, or static FNNs that implement the same I/O (or approximate it arbitrarily well) \cite{doerig2019unfolding}. This might include potentially more exotic options, like substitutions based on the AIXI model of universal intelligence \cite{hutter2000theory} or even cases of a brain in quantum superposition (a ``Schrödinger's Zombie'') \cite{brown2019schrodingers}.
    \item Algorithmic substitutions accomplish any particular I/O function via a different algorithm or program with different workings. For instance, there are various ways to sort (bubble sort, merge, etc.) that can all transition an unsorted list to a sorted one.
    
\end{itemize}

It had been hoped that issues of falsification like the Unfolding Argument would only ever impact ``causal structure'' theories like IIT \cite{doerig2019unfolding}. However, once generalized as the Substitution Argument, there did not seem any particular protection for functionalist or computational theories \cite{kleiner2021falsification}. This is likely because the definition of what is functionalist and what is causal structure is difficult to define (e.g., functionalism may require a ``causal topology'' \cite{chalmers2011computational}). In what follows I show precisely how the Kleiner-Hoel dilemma impacts functionalist and computational theories of consciousness too. Instead of substituting systems, we can substitute programs. Specifically, the shortest program.

\begin{theorem}[Kolmogorov Substitution Theorem]

Following \cite{kolmogorov1965three}, for any computable function $f$ there exists a shortest program $K(f)$ implementing $f$ (even if it is not known how to identify it). This $K(f)$ can always be used as a universal substitution for any other program $N$ implementing $f$, and so will cause mismatches and falsify a given theory of consciousness that targets the program $N$ as the basis of its predictions. Alternatively, if a theory of consciousness has the same predictions for all programs implementing $f$, including $K(f)$, then the predictions are entirely determined by $f$ itself, not the original program implementing it, since $K(f)$ is solely a function of $f$. In that case, predictions and inferences have the same source and cannot vary, leading to triviality.

\end{theorem}

\begin{proof}
\textbf{Step 1.} For any computable $f$, $K(f)$, the shortest program implementing $f$, exists, and is itself solely a function of $f$.

\textbf{Step 2.} Let $N$ be any program computing $f$ that is the target of the predictions of some theory of consciousness. $N$ and $K(f)$ share I/O (defined by $f$).

\textbf{Step 3.} Any theory that assigns different predictions about consciousness to $N$ (the actual program under question) and $K(f)$ (its universal substitution) is falsified, since this implies empirical mismatches.

\textbf{Step 4.} Any theory that always shares predictions between $N$ and $K(f)$ is basing those predictions on properties of $K(f)$ alone. Since predictions based on $K(f)$ can be derived entirely without any knowledge of program $N$, this renders program $N$ irrelevant to the given theory of consciousness. Therefore, any such theory collapses into triviality.
\end{proof}

Many theories could be shown to have problems with falsification in this light. Consider just one: Higher Order Thought (HOT) theory, a popular theory of consciousness \cite{brown2019understanding}. Any exact ranking or rating of the different representations should differ in its predictions about $N$ vs. $K(f)$. But if somehow not, then $N$ is irrelevant and only the I/O and its subsequently determined $K(f)$ is needed for the theory's predictions, rendering it trivial. The same goes for almost all theories of consciousness that are more focused on functionalist aspects, like ``representation'' or ``computation.'' As a concrete example, it has been argued that ``Perceiver'' architecture of a neural network satisfies the properties of Global Workspace Theory \cite{juliani2022perceiver, dehaene1998neuronal}. This architecture could be replaced by a lookup table, a single-hidden-layer FNN, or $K(f)$, all while holding I/O constant, and all involving substantial variations in predictions (such as a lack of consciousness).

The overall structure of this proof also demonstrates how the Kleiner-Hoel dilemma has two horns. The first horn is the Substitution Argument: If substituting in the shortest program that computes the function $f$ makes a theory's predictions change, then it creates mismatches, enacting a kind of ``a priori falsification.'' However, on the other hand, if the predictions of a theory of consciousness do \textit{not} change, then they must be strictly dependent on the experimental inferences observed during testing a theory of consciousness, rendering that theory trivial (here, this can be seen by any details of the original program $N$ \textit{no longer mattering at all} if a theory's predictions are always the same for $K(f)$ too). This would be a case of ``strict dependency'' wherein predictions either mimic, stem from, or can be entirely derived by the inferences (in the terminology of Theorem 3.1, inferences would be entirely determined by $f$).

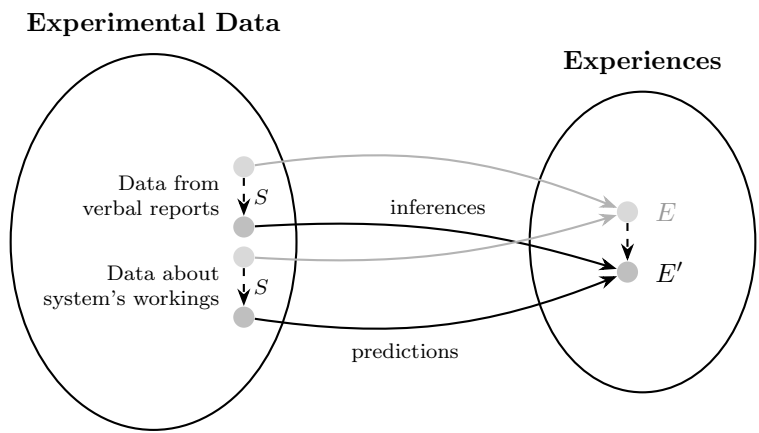
\begin{figure}[ht]
\centering
\begin{tikzpicture}[
    >=Stealth,
    every node/.style={font=\small},
    set/.style={ellipse, draw, thick, minimum width=3.8cm, minimum height=5cm, align=center},
    element/.style={circle, fill=black!25, inner sep=0pt, minimum size=8pt},
]
\node[set] (data) at (0,0) {};
\node[font=\bfseries, above=0.1cm of data.north] {Experimental Data};
\node[set, minimum width=3cm, minimum height=4cm] (exp) at (6.5,0) {};
\node[font=\bfseries, above=0.1cm of exp.north] {Experiences};
\node[element, fill=black!15] (reports_old) at (1.2, 1.0) {};
\node[element] (reports_new) at (1.2, 0.2) {};
\node[anchor=east, align=right, font=\footnotesize] at (1.0, 0.6) {Data from\\verbal reports};
\draw[->, thick, dashed] (reports_old) -- (reports_new) node[midway, right, font=\footnotesize] {$S$};
\node[element, fill=black!15] (neural_old) at (1.2, -0.2) {};
\node[element] (neural_new) at (1.2, -1.0) {};
\node[anchor=east, align=right, font=\footnotesize] at (1.0, -0.6) {Data about\\system's workings};
\draw[->, thick, dashed] (neural_old) -- (neural_new) node[midway, right, font=\footnotesize] {$S$};
\node[element, fill=black!15] (E_old) at (6.3, 0.4) {};
\node[anchor=west, font=\itshape, text=black!40] at (6.55, 0.4) {$E$};
\node[element] (E_new) at (6.3, -0.4) {};
\node[anchor=west, font=\itshape] at (6.55, -0.4) {$E'$};
\draw[->, thick, draw=black!30] (reports_old) to[bend left=15] 
    (E_old);
\draw[->, thick] (reports_new) to[bend left=10] 
    node[midway, above=2pt, font=\footnotesize] {inferences} 
    (E_new);
\draw[->, thick, draw=black!30] (neural_old) to[bend right=10] 
    (E_old);
\draw[->, thick] (neural_new) to[bend right=15] 
    node[pos=0.4, below=2pt, font=\footnotesize] {predictions} 
    (E_new);
\draw[->, thick, dashed] (E_old) -- (E_new);

\end{tikzpicture}
\caption{Strict dependency: when inference and prediction data are strictly dependent, any substitution $S$ changes both simultaneously to the same degree. Since inferred and predicted experiences always shift together, they always match, rendering a theory unfalsifiable.}
\label{fig:strict-dependency}
\end{figure}

Here and throughout, such strictly-dependent-on-inferences theories of consciousness are referred to as ``trivial.'' E.g., a widely-held (but trivial) theory of consciousness is behaviorism, such that if an entity acts conscious, then it is conscious. But such a theory cannot be falsified, since its predictions could never be any different than experimental inferences (as they are simply a restatement of them) and so such a theory contains no scientific information, and so it is trivial. While the term ``functionalism'' has many different meanings, for the same reason that behaviorist theories of consciousness are trivial, so is any theory of consciousness based \textit{solely} on I/O. Arguably, this does significantly affect plenty of existing theories, which turn out to be simply behaviorist theories with multiple steps. E.g., when discussing Global Neuronal Workspace Theory (GNWT), Daniel Dennett once wrote that:

``... theorists must resist the temptation to see global accessibility as the \textit{cause} of consciousness (as if consciousness were some other, further condition); rather, it \textit{is} consciousness'' \cite{dennett2001we}.

If Dennett's ideas were true, this would make GNWT a trivial theory of consciousness, since it would be unfalsifiable, because the predictions and inferences strictly depend on the same source. Consciousness just \textit{is} the information that is globally accessible for report and behavior, and reports and behaviors also just \textit{are} different expressions of that same information. So there is no situation wherein they could diverge. And therefore, the theory is unfalsifiable and trivial (one could see how it gives us minimal scientific information even if it were true---trivial theories are fundamentally unsatisfying). 

Often, theories based in the ``theater'' metaphor have a tendency to fall into strict dependency (sometimes ``global workspaces'' are described in this way too \cite{baars1997theater}). Theories based on identifying only why humans might \textit{claim} to be conscious (rather than why they \textit{are} conscious) are at high risk of triviality (e.g., theories like the Attention Schema Theory of consciousness \cite{graziano2020consciousness}, or any other theory focused on the ``meta-problem of consciousness'' as an explanandum \cite{chalmers2018meta}).

A significant amount of confusion about the Kleiner-Hoel dilemma was simply missing that the dilemma has this second horn of strict dependency. E.g., in one of the only cases in the literature of attempting to confront the problem head-on, the criticism ended up being that the substitution argument could not possibly be true, because it did not factor in needing to change the predictions for functionalist theories of consciousness after substitutions \cite{ganesh2020cwars}. But in arguing that predictions must change just the same way as inferences following a substitution, the paper proposes a trivial theory of consciousness (a kind of global I/O functionalism) that falls on the second horn of the Kleiner-Hoel dilemma.

In the next section, I show that the Kleiner-Hoel dilemma provides a particularly strong disproof of LLM consciousness. After that, I present a case of avoiding the dilemma in humans but not LLMs, indicating that it disproves consciousness in LLMs but not in humans.

\section{The Proximity Argument}

\begin{definition}[Trivial Theories]
A theory of consciousness, $T$, is trivial if there is strict dependency between its predictions and inferences. That is, if the prediction function, $pred$, and inference function, $inf$, are identical (or if the two are so approximate there can only be minimal scientific information gained for any possible experiment). An example of triviality is if $inf$ cannot be varied without $pred$ also being varied, or both stem from the same source of data, then such a theory is unfalsifiable since mismatches are impossible.
\end{definition}

Here, I assume that trivial theories of consciousness must be false (rather than just trivial), since otherwise there is no scientifically informative theory of consciousness. This assumption makes the language and conclusions easier to follow throughout (e.g., instead of writing ``a system for which there are no non-trivial theories of consciousness available'' I can say ``a non-conscious system''); however, all the conclusions would still hold but with this sort of replacement language if this assumption (which is just that a scientifically-informative theory of consciousness actually exists) were not true. See subsection 4.1 for relevant discussion about rejecting this and instead accepting trivial theories as true. 

\begin{definition}[Non-conscious Systems]
Following Definition 4.1 of Trivial Theories: A system $s$ is a non-conscious system if, for any theory of consciousness $T$, whenever $T$ predicts $s$ is conscious, $T$ must necessarily be trivial.
\end{definition}

A lookup table offers an archetypal example of a non-conscious system under Definition 4.2. Any lookup table $L$ simply implements an I/O function $f: I \rightarrow O$. Imagine a theory $T$ that predicts $L$ has non-trivial consciousness. What could $T$ base its prediction of consciousness on? $L$ is not internally complex, lacks all dynamics, is entirely static, has no memory capabilities, and there is no information flow across any parts or modules. Certainly, a theory cannot make meaningful predictions of consciousness based on merely the fact that some set of IF-THEN statements exist. The sole relevant property for a theory of consciousness pertaining to $L$ is the I/O of the function $f$ itself. However, that's the same function that determines all the experimental inference data! Therefore, any $T$ predicting $L$ is conscious must base this prediction solely on $f$, which is identical to the basis of inference, creating strict dependency. Thus any given $T$ is trivial. Thus, $L$ is \textit{necessarily} a non-conscious system (as long as, per Definition 4.1, non-trivial theories of consciousness exist). 

\begin{definition}[Substitution Distance]
For a pair of systems, $s$ and $s'$, that share identical I/O (or approximately equivalent I/O, wherein this could describe behavior or responding to text, etc.), the substitution distance $d_S(s, s')$ is determined by the space of properties $\mathcal{P}$ that differ between $s$ and $s'$.
\end{definition}

As an example of Definition 4.3, consider the well-detailed case of unfolding a recurrent artificial neural network (RNN) into a feedforward artificial neural network (FNN) that implements the same I/O function, as in \cite{doerig2019unfolding}. In such an unfolding, many properties are lost. E.g., while a single-hidden-layer FNN can approximate any continuous function to arbitrary precision \cite{cybenko1989approximation, hornik1989multilayer}, it lacks recurrence, temporal dynamics, higher-order structure, broadcasting, depth, etc. On the other hand, some things remain unchanged, like its substrate and activation functions (i.e., it is still an artificial neural network with individual artificial neurons and connections, running on a computer, etc.). The scope of these lost properties defines the \textit{substitution distance} (Definition 4.3).

\begin{definition}[Consciousness-Relevant Properties]
For a system $s$ with substitution distance $d_S(s, s')$ to a non-conscious system $s'$, let $\mathcal{P}$ be the set of properties lost following this substitution. A theory of consciousness, $T$, is non-trivial only if $T$'s predictions are sensitive to (or grounded in) some subset of $\mathcal{P}$, since by Definition 4.2, for non-conscious systems only trivial theories are possible. 
\end{definition}

This follows naturally from the previous definitions: given a pair of systems, and one is a non-conscious substitution, then for a theory of consciousness to be non-trivial it must somehow differentiate the two. If it is a non-trivial theory of consciousness, it cannot predict that the non-conscious system is conscious, so therefore the difference must lie in the properties lost compared to the non-conscious substitution.

\begin{proposition}[Proximity Principle]
The larger the substitution distance $d_S(s, s')$ between a system $s$ and a non-conscious system $s'$, the larger the space of properties $\mathcal{P}$ that could (at least potentially) ground non-trivial predictions. Conversely, if $d_S(s, s')$ is small, and so if $s$ and $s'$ differ in only a few properties, then there is a small space of properties that could ground non-trivial predictions. If none of the differing properties can ground non-trivial predictions, then $s$ is also a non-conscious system.
\end{proposition}

According to the Proximity Principle, given two different systems, one with a large ``substitution distance'' from a non-conscious system, and the other with a small ``substitution distance'' from a non-conscious system, then in the latter case there are much fewer properties that could \textit{make a difference} when it comes to predictions of consciousness.

\begin{theorem}[Constraint Theorem]
Any non-trivial theory of consciousness that predicts contemporary LLMs are conscious must base its predictions entirely on properties constrained to the substitution distance between a lookup table and a static single-hidden-layer FNN, combined with the substitution distance between a static single-hidden-layer FNN and an LLM.
\end{theorem}
\begin{proof}
\textbf{Step 1.} Let $M$ be a deployed static LLM, let $N$ be a static single-hidden-layer feedforward neural network implementing the same I/O function as $M$, and let $L$ be a lookup table also implementing the same I/O function.

\textbf{Step 2.} By the universal approximation theorem, any I/O function (like text responses) computed by $M$ can be approximated arbitrarily well by a sufficiently-wide static single-hidden-layer network $N$. By definition, any finite function computed by $N$ can in turn also be implemented as a lookup table, $L$.

\textbf{Step 3.} Therefore there exists a chain of universal substitutions preserving I/O:
\[L \leftrightarrow N \leftrightarrow M\]

\textbf{Step 4.} By Definition 4.2, $L$ is a non-conscious system. By the Proximity Principle (Proposition 4.5), any non-trivial theory $T$ predicting $M$ is conscious must base its predictions solely on properties within the space defined by $d_S(M, L) \subseteq d_S(M, N) \cup d_S(N, L)$.
\end{proof}

\begin{corollary}
If no property within this space constrained by this universal substitution chain can ground non-trivial predictions for a theory of consciousness, then LLMs are non-conscious systems.
\end{corollary}

At this point in our reasoning, we should already have strong reasons to question LLM consciousness, simply due to how significant this constraint is. Certainly, there are some details like scaffolding, interfaces, if it is deployed, and what type of contemporary LLM, etc., that I am leaving aside. However, at least when it comes to the core ``bare bones'' of a baseline contemporary LLM (like, say, the original GPT-4, with no router or other complications), any claim to consciousness is confined to its distance to a static single-hidden-layer approximation of its I/O (which is guaranteed by the universal approximation theorem \cite{hornik1989multilayer}). And note: Here, the substitution of a static single-hidden-layer FNN is another kind of artificial neural network that computes I/O in much the same way, in that it has similar activation functions, is also instantiated on a computer, has identical I/O, and so on. So if we just stopped here at Corollary 4.7, LLM consciousness would, at minimum, be highly constrained. It must require something lost in that chain of universal substitutions (like, say, extra layers). 

If one still holds that a single-hidden-layer static FNN is non-trivially conscious, then there is also the distance to the provably-trivial lookup table to consider (there, the only relevant property appears to be compression, which is also not a good target for theories of consciousness---see Corollary 4.10 for a proof that static single-hidden-layer static FNNs are non-conscious). If a theory of consciousness cannot specify what is important and lost along this substitution chain, then it cannot claim that LLMs are conscious.

However, one can go further, because this so far has only discussed the second horn of the Kleiner-Hoel dilemma: the relationship between triviality and strict dependency. If we add back in the first horn of mismatches following substitutions to this transitive chain, the proof rules out baseline LLM consciousness entirely. 

\begin{proposition}
Any theory of consciousness grounded in properties within $d_S(N, L)$ or $d_S(M, N)$ is either trivial or is a priori falsified by the substitution chain $L \leftrightarrow N \leftrightarrow M$.
\end{proposition}

\begin{proof}

\textbf{Step 1.} The substitution chain $L \leftrightarrow N \leftrightarrow M$ is constructed to preserve I/O. Since empirical inferences about consciousness for entities like LLMs are derived from I/O functions (not from, e.g., how many layers they possess), all inference data is fixed across the chain.

\textbf{Step 2.} Let $\rho$ be some specific property in $d_S(M, L) \subseteq d_S(M, N) \cup d_S(N, L)$. By Definition 4.3 of Substitution Distance, assume that $\rho$  varies across the chain.

\textbf{Step 3.} For any theory $T$ that bases its predictions on $\rho$: since $\rho$ varies while inference data is by definition fixed, predictions are necessarily independent from inferences with respect to $\rho$. This necessarily leads to mismatches, meaning that $T$ is a priori falsified.

\textbf{Step 4.} The only alternative is if $\rho$ does not vary, in which case, $T$ must necessarily ground predictions solely in the I/O structure itself. This creates a situation of strict dependency, rendering $T$ trivial.

\textbf{Step 5.} Therefore, for any given theory of consciousness $T$ it either: (a) grants consciousness to contemporary baseline LLMs, but does not for trivially non-conscious systems, yet then $T$ is a priori falsified by mismatches in prediction when non-conscious systems are substituted in for LLMs; or (b) grants consciousness to LLMs and non-conscious systems, but then the theory must apply in the latter case, and so must be trivial.
\end{proof}

An alternative summary of Proposition 4.8 is to say that the substitution chain $L \leftrightarrow N \leftrightarrow M$ is constructed to fix I/O, which essentially exhausts all relevant inference data (especially in light of Corollary 4.10, which proves static single-hidden-layer FNNs are non-conscious), since there would have to be a priori neutral and non-empirical and non-theory-based reasons to assume LLM consciousness, which somehow do not apply to single-hidden-layer FNN consciousness, and to such a degree that it (somehow) overrides the pathological mismatches with a non-conscious system. Therefore, any property that varies across the chain should be independent of inferences. This eliminates the possibility of overcoming the Kleiner-Hoel dilemma within this substitution chain.

\begin{theorem}[Disproof of LLM Consciousness]
If Proposition 4.8 holds, then deployed baseline LLMs are not conscious.
\end{theorem}

\begin{proof}
By the Constraint Theorem, any non-trivial theory predicting LLMs are conscious must base its predictions on properties in $d_S(M, N) \cup d_S(N, L)$. By Proposition 4.8, every property in this space is falsified by the universal substitution chain itself. Therefore, under Definition 4.2, LLMs are non-conscious systems.
\end{proof}

One thing that's notable here is that Theorem 4.6 alone, just by strict dependency requirements, significantly constrains what kinds of theories of consciousness can apply to LLMs to some subset of their difference from non-conscious systems. It says that a theory of consciousness must be grounded in some property (or properties) lost between the substitution of an LLM to a (functionally equivalent) single-layer feedforward neural network to a (functionally equivalent) lookup table. Based on this restricted space of possibilities, Theorem 4.9 disproves their consciousness entirely by noticing that the fact that a universal substitution chain exists means that, even if there were some property lost (that a hypothetical theory of consciousness picked out) that only means mismatches would occur and so the theory would be a priori falsified. 

An example of how the Proximity Argument works can be seen by zooming into one link in the chain from lookup tables to single-hidden-layer FNNs. Specifically, we can make the same form of argument there, using a specific property, compression, that is within the substitution distance (and the only one available substantial enough to ground a theory of consciousness).

\begin{corollary}[\textbf{Static Single-hidden-layer FNNs are non-conscious systems}]
Examining a single step in the Proximity Principle: consider the substitution distance $d_S(N, L)$ between a lookup table, $L$, and a static single-hidden-layer FNN, $N$, implementing the same I/O function. This distance can be entirely summarized as the degree of compression of the I/O function, and since compression can be varied via universal substitutions, therefore single-hidden-layer FNNs are non-conscious systems.

\begin{proof}
\textbf{Step 1.} A lookup table with $n$ discrete inputs can be implemented as a maximally wide static single-hidden-layer FNN (with one hidden unit per possible I/O).

\textbf{Step 2.} Substitute in a narrower network with identical I/O. Since that I/O is fixed, the only difference is in the degree of compression in the implementation.

\textbf{Step 3.} Any theory $T$ that grounds its predictions of consciousness in compression can therefore yield independent predictions across networks, all while inferences (necessarily derived from I/O) remain fixed---thus creating the pathological mismatches that falsify $T$.

\textbf{Step 4.} Therefore, only the I/O function itself remains to actually ground $T$, and so by Definition 4.2 of non-conscious systems, static single-hidden-layer FNNs are non-conscious systems.
\end{proof}

\end{corollary}

Regarding Corollary 4.10, it is notable that $K(f)$ could be used instead, representing the point of maximal compression (while lookup tables represent the point of minimal compression).

Zooming out: there exist in the literature previous ``no-go'' results on AI consciousness, such as arguments against AI consciousness based on causal exclusion \cite{kleiner2024dynamically} (though similar arguments have been given against the efficacy of human consciousness as well \cite{kim2000mind}), as well as other arguments for biological naturalism \cite{seth2024conscious, campero2025consciousness}. However, one unique aspect of the Proximity Argument is that it makes minimal assumptions other than that: (a) a theory of consciousness cannot be trivial or a priori falsified, (b) that certain universal substitutions are possible (which we know from existent mathematical theorems to be true), and (c) that of widespread or universal mismatches from a theory in the comparative testing setup described would falsify that theory. Nor does it rely on a particular theory; e.g., while transformers provably are feedforward and have zero integrated information (thus, IIT would say that LLMs are not conscious) \cite{ali2025intelligence}, the Proximity Argument goes beyond relying on a particular theory. It requires no opinion on whether qualia is causally relevant or epiphenomenal, nor whether functionalism is true, nor any other philosophical position about consciousness, nor does it rely on any particular theory of consciousness or even class of theories to be true (it does not, e.g., require solving the Hard Problem \cite{chalmers1995facing}). It is a philosophically neutral proof that both constrains and disproves LLM consciousness based on falsification, linking up to historical critiques about functionalism and intelligence (e.g., the Blockhead thought experiment \cite{block1981psychologism} and the Chinese Room thought experiment \cite{searle1980minds}) although it is not framed as a \textit{reductio ad absurdum} but based on concerns about falsifiability.

It should also be noted that, following Theorem 3.1, aspects of the Proximity Argument could be restated in terms of the shortest program $K(f)$ acting as the substitute for the LLM. Indeed, the shortest program is much like a mirror of a lookup table: instead of being the maximally-complex implementation (wherein every part relates to the function, at least), the shortest program is the minimally-complex implementation, and it too is (a) different from whatever original program or system was specified, and (b) just as much a direct function of I/O.

\subsection{The Kleiner-Hoel Dilemma is Particularly Forceful for LLMs}

Why does the Kleiner-Hoel dilemma not apply just as much to humans? Here, I give multiple reasons why the arguments are particularly forceful in LLMs while considering possible interpretations of these proofs. I also consider cases where the Kleiner-Hoel dilemma \textit{does} apply in humans, how theories might escape the dilemma, and how rescuing LLM consciousness requires extreme and unscientific beliefs about consciousness.

\begin{enumerate}
    
    \item \textbf{Non-conscious universal substitutions for LLMs are well-defined and constructible.} In some cases, universal substitutions may be merely conceivable (e.g., consider an infinite lookup table implementing $f$, wherein $f$ is somehow given). However, in practice, this apparent conceivability may turn out to be ill-defined. E.g., can a human brain be substituted with a functionally-equivalent single-hidden-layer FNN, even one impractically astronomically wide? Is a human's \textit{f} clearly definable? Asking these questions explicitly is \textit{not} necessarily agreeing that universal substitutions are ill-defined in humans---indeed, the entire rest of the paper will assume these questions are answerable. However, it is undeniable there is significant subtlety and complexity to them. It may be that actual definable and constructible universal substitutions for humans (especially those that count as non-conscious according to Definition 4.2) could be somehow ruled out by details we don't know or are hard to foresee (e.g., quantum processes in the brain \cite{hameroff2014consciousness}, or some computational irreducibility in the process \cite{wolfram2002new}, or the no-cloning theorem \cite{Wootters1982} or some other factor that crops up in an unexpected way).
    
However, for baseline LLMs, it appears from the literature that universal substitutions are \textit{actually} constructible via chaining together known methods. Specifically, it is known that transformers can be arbitrarily approximated as RNNs \cite{katharopoulos2020transformers, choromanski2021rethinking}. Then, it is known that (static) RNNs can be ``unfolded'' into a single-hidden-layer FNN \cite{doerig2019unfolding, o2025caveat}, which can in turn be replaced via enumeration via a lookup table. So the chain discussed in the Proximity Argument (e.g., Theorem 4.6) is literally constructible via known methods, and could potentially even be given as defined transformations. Perhaps for a small-enough language model, some version of it could be transformed this way with current technology. That is, we seem to have a general understanding of how to implement non-conscious (Definition 4.2) universal substitutions for LLMs. This should call into question their consciousness significantly. 
    
    \item \textbf{The potential for theories that evade the Kleiner-Hoel dilemma in humans.} Due to the Proximity Principle (Proposition 4.5) there is a far greater Substitution Distance (Definition 4.3) between humans and available non-conscious substitutions (which, from now on, I assume are well-defined and possible). This means that humans have more Consciousness-Relevant Properties (Definition 4.4). Meanwhile, LLMs are much more constrained via the Constraint Theorem (Theorem 4.6). This means there are far more options for theories of consciousness in humans compared to LLMs. Importantly, this means there are far more potential theories that demonstrate ``lenient dependency'' in that they avoid the Kleiner-Hoel dilemma but still are falsifiable and contain meaningful scientific information.
    \begin{definition}[Lenient Dependency]
A theory of consciousness $T$ satisfies lenient dependency if there is not strict dependency between its predictions and inferences and it has no pathological mismatches in predictions due to definable possible universal substitutions.
    \end{definition}

    The remaining sections of this paper will deal with exploring the nature of this lenient dependency via a class of theories of consciousness that are sensitive to or grounded in properties of continual learning.

    \item \textbf{The greater potential for theories that are unscientific or unfalsifiable to assign consciousness to humans.} The alternative to rejecting trivial theories and thus denying conclusions like Definition 4.2 of Non-conscious Systems, is to accept trivial theories (and only trivial theories; otherwise, they would be falsified by the first horn of the Kleiner-Hoel dilemma). And there are indeed various theories of consciousness that may contain minimal scientific information and so avoid the Kleiner-Hoel dilemma by being trivial. These might be categorized as metaphysical theories that are also scientifically trivial, as per Definition 4.1 of trivial theories. However, instead of rejecting such theories in favor of a scientifically testable one, as here, one might imagine accepting them. Even so, all such theories have far more ``room'' in humans, due to the many more properties of humans that can serve as Consciousness-Relevant Properties (Definition 4.4). It is beyond the scope of this paper to delineate metaphysical theories that are still scientifically trivial per Definition 4.1; however, potential candidates include theories like substance dualism \cite{descartes1984meditations} or analytic idealism \cite{kastrup2019idea}. Even panpsychist theories like Russellian monism \cite{chalmers2015panpsychism} face a combination problem \cite{morch2019integrated} that seems more solvable in an integrated and plastic human brain than a feedforward and static LLM. So even when advocating for a trivial theory of consciousness, belief in LLM consciousness is constrained to a subset of theories which seem the least metaphysically ``worthwhile'' (e.g., behaviorism).

    \item \textbf{The scope of inferences could be restricted.} One potential strategy for dealing with the Kleiner-Hoel dilemma is to severely restrict the scope of consciousness science. It has been advocated to simply test theories in humans and ignore contradictory reports from substitutions \cite{albantakis2020unfolding, usher2023philosophical}. It is imagineable that this could be justifiable if consciousness science is unlike other sciences and relies on first-person evidence  \cite{tononi2025consciousness}. In particular, this fits with theories based in phenomenology, like IIT \cite{oizumi2014phenomenology}. However, it is already known that such theories show that LLMs would possess zero integrated information \cite{ali2025intelligence}, and IIT would rule out even a simulated brain on a computer being conscious \cite{findlay2024dissociating}. Overall, while this seems conceptually possible, this move is restrictive. A theory of consciousness would be not universal or robust, but highly fragile to falsification, and kept from falsification via artificial restrictions. Currently, it also remains undefined. What, ultimately, is the difference between an experiment and a substitution? Already there have been proposed experiments in humans that complexify any attempted demarcation \cite{dennett2001we}. And most relevantly to the thesis here, restricting inferences is particularly problematic if someone were to advocate for it and also LLM consciousness at the same time, as it means accepting that non-human behavior (even highly complex) contains zero actionable scientific information about consciousness. Therefore, there would be no particular reason to think LLMs would possess consciousness to begin with, and no behavior from an LLM would count as evidence for its consciousness.
    
    \item \textbf{The Kleiner-Hoel dilemma could indicate the undecidability between theories of consciousness.} I have previously argued that consciousness could act much like a Gödel sentence in science \cite{hoel2024world}---there are alternative proposals for various forms of related mysterianism as well \cite{mcginn1999mysterious}. This would likely mean there are no falsifiable theories of consciousness, and, if the Kleiner-Hoel dilemma was truly universal across systems (including humans) and theories of consciousness (and \textit{only} if this universality held) it would support that consciousness is beyond empirical science's ken; however, the reasoning for metaphysical but unscientific trivial theories having more ``room'' for humans in them would still apply (as in \#2).
    
\end{enumerate}

For the rest of the paper, I explore the scientifically-promising initial given reason for why the Kleiner-Hoel dilemma may not apply in humans: theories of consciousness may possess lenient dependency.  In the following section, I argue that lenient dependency must be a continually-existing criterion for theories of consciousness (a high bar). Then, I propose that learning is a form of lenient dependency in humans, and that this is necessarily restricted to truly \textit{continual} learning (the learning occurring with, or potentially even constituting, every single conscious experience). If this further argument is correct, it would significantly constrain consciousness in humans to involve the process of continual learning. I argue that this restriction is scientifically positive: while it rules out most existing theories in the field, it still points us toward evocative and rich theories of consciousness and reveals what a minimal falsifiable and non-trivial theory of consciousness would look like. Additionally, it indicates that a fundamental difference between LLMs and humans is the ability to continually learn and that LLM's lack of consciousness may be tied to their static nature.

\section{Continual Lenient Dependency}

Here, I argue that avoiding the Kleiner-Hoel dilemma forms a strong constraint on theories, in that theories must \textit{continually} be based on properties that satisfy lenient dependency at \textit{every prediction} by the theory with regard to the contents or presence of consciousness. In other words, lenient dependency is a continual and ongoing requirement for theories of consciousness.

\begin{proposition}[Continual Requirement for Lenient Dependency]
Let $T$ be a theory of consciousness with lenient dependency grounded in some property $p$. Then $p$ must be continually present during the detection of consciousness: that is, there can be no temporal gaps during which a system $s$ lacks $p$ while $T$ continues to make non-trivial predictions that $s$ is conscious.
\end{proposition}

\begin{proof}
\textbf{Step 1.} Assume lenient dependency for $T$ is grounded in property $p$ of system $s$.

\textbf{Step 2.} Suppose there exists a gap: a time $t^*$ at which $s$ lacks $p$, yet $T$ still predicts that $s$ is conscious.

\textbf{Step 3.} At $t^*$, without $p$, system $s$ can be substituted with a non-conscious system $s'$ (Definition 4.2), such as a static single-hidden-layer FNN (Corollary 4.10).

\textbf{Step 4.} Therefore, $T$ must either (a) make different predictions for $s'$, creating pathological mismatches, or (b) predict consciousness in $s'$, indicating that $T$ itself is trivial.

\end{proof}

The only exception to this would be if \textit{both} (a) a theory could be tested only under conditions of lenient dependency and those instances are enough scientific weight to carry a theory, and \textit{also} (b) the theory were grounded in different properties at different times. However, this seems to indicate two separable theories of consciousness, and the theory that was not grounded in continual lenient dependency would fail for the same reasons as in Proposition 5.1.

\begin{corollary}
For a theory of consciousness, the property $p$ grounding lenient dependency must be present whenever the theory predicts consciousness; otherwise, the theory is necessarily trivial (strict dependency) or a priori falsified by universal substitutions.
\end{corollary}

That is, for any case of lenient dependency, it must be constantly ongoing during predictions of consciousness. This significantly constrains theories of consciousness. It also makes it very clear that it is not the kind of system being tested that matters (e.g., if a system is plastic, or reentrant, etc.) but rather the properties a theory of consciousness is itself grounded in \textit{all the time}.

\subsection{Continual Learning as Continual Lenient Dependency}

Lenient dependency requires that a theory's predictions not be a priori falsified by known universal substitutions but is \textit{also} not strictly dependent, and we now know from Proposition 5.1 that this must continually hold or else a theory of consciousness collapses. Here, I show that a theory involving continual learning offers such a hypothesis by significantly complicating the viability of universal substitutions.

Naturally, a theory of consciousness involving learning would involve testing learning rather than merely static I/O (including evidence of learning in inferences). Specifically, we can define for some static function \textit{f}, there is $\Delta f$, which here represents the change from learning.

\begin{definition}[Static System]
For a system $s$ with an I/O function, $f$, then $\Delta f = 0$ if identical-in-content inputs performed at different times cannot possibly yield different outputs, excluding randomness or internal changes not caused by inputs.
\end{definition}

\begin{proposition}
Static Systems (Definition 5.3), such as a lookup table $L$ implementing the function $f$, would be unable to act as a universal substitution for a learning system in a consciousness-testing regime sensitive to $\Delta f$, since for the original system $\Delta f \neq 0$ and by Definition 5.3 $L$'s $\Delta f = 0$.
\end{proposition}

\begin{corollary}
Static Systems which implement a static $f$ cannot be valid universal substitutes for systems that learn and for which $\Delta f \neq 0$. Any universal substitute must either actually itself learn as well (via things like changes to internal memory), or be a non-learning system which either:

\begin{enumerate}[(a)]

    \item  \textbf{Stores external memory.} A system may input $(x, \text{history})$ rather than $x$ alone, like a static lookup table. This would violate Definition 5.3 of being a Static System, since identical inputs cannot yield differentiable outputs. Additionally, I/O will necessarily expand in scope to include another kind of data (history) rather than, e.g., sensory processing, leading to different domains over the input (and possibly output if a loop is necessary), indicating a failure to substitute.
    
    \item  \textbf{Is experimentally derivative. }A series of static substitutions might replace the trajectory of a learning system. However, testing would require ongoing reconstruction, i.e., the series cannot be truly substituted in for the original and experimented on in isolation, as the series depends on continually monitoring changes in the original, upon which the actual experiments that include $\Delta f$ must also be performed. 

\end{enumerate}
\end{corollary}

This is supported by the literature: indeed, a specific case that a plastic recurrent neural network actually cannot be unfolded into an equivalent static FNN has been shown \cite{o2025caveat}. However, by focusing on learning rather than on plasticity, and substitutions in general rather than just unfolded networks, this offers a broader case. E.g., in \cite{o2025caveat} many of the proofs are based on resource constraints, and do not rule out cases like an infinite static system that contain all possible histories acting as a viable substitution (which are physically impossible, but still potentially relevant to reasoning about falsification). However, Corollary 5.5 indicates that even such instances are impossible. Note that this applies to, e.g., LLMs acting as substitutes for humans: for an LLM, the context window is fed back into the input, thus, via Corollary 5.5(b), baseline LLMs mimic learning via external memory. Fed the same $\text{input}(x,\text{history})$ at any time, their output probabilities would be identical. Additionally, their I/O is actually \textit{quite unlike} humans even if they were to act as substitutes in, e.g., a Turing Test \cite{turing2021computing} (each query actually involves the entire previous conversation fed into a static system). This covers the case for what is often referred to as ``in-context learning'' for LLMs.

Given Proposition 5.4's truth, the failure of all static systems to offer viable universal substitutions for learning systems means that the set of universal substitutions is a smaller subset when it comes to theories of consciousness that involve learning. In other words, theories of consciousness that involve learning are harder to a priori falsify by finding universal substitutions. Examples of substitutions in the literature have relied entirely on known mathematical methods like the universal approximation theorem \cite{doerig2019unfolding, hornik1989multilayer} to arbitrarily approximate some continuous function $f$, not some $\Delta f$.  This is not to say that the results herein currently show that \textit{all} universal substitutions for learning systems are impossible. However, they do rule out the class of static substitutions for theories of consciousness within which predictions are grounded in or require continual learning, and point toward testing necessarily including learning. 

Learning also offers clear ways to further complexify the issue that seem to make substitutions harder to define or implement: for instance, learning involves generalization to new never-before-seen inputs, but the status of which under question for I/O substitutes---can you really substitute \textit{every} possible history once learning is involved, without necessarily relying on infinities? Additionally, learning systems may not be enumerable in the same way as static ones, since enumeration of inputs and outputs leads to plastic changes over those very inputs and outputs. Indeed, focusing on learning-based theories of consciousness may connect to the hypothesis that consciousness may involve what Geoffrey Hinton dubbed ``mortal computation'' (for a line of argument that consciousness must be such a mortal computation, see \cite{kleiner2024consciousness})---defined as computation unable to be divorced from its substrate \cite{hinton2022forward}. With regards to continual learning, there could be an analogous version of ``mortal learning,'' wherein changes caused by input depended on no static memory being read/written. If so, this would rule out any universal substitution that could be viewed (under some perspective) as undergoing separable read/write execution as part of the input of the system, thus changing the scope of I/O, for similar reasons as in external memory in Corollary 5.5, and thus failing to be a universal substitution for a mortal learning system.

Finally, we can ask whether (and why) theories of consciousness focused on the process of continual learning would not be strictly dependent. This is much easier to show, since to hold it simply has to be the case that \textit{pred} and \textit{inf} can vary or don't stem from the exact same data. This seems likely the case due to things like the multiple-realizability of learning, or cases of latent learning, and so on.

\begin{proposition}[Learning Escapes Strict Dependence]
A theory of consciousness grounded in continual learning is not strictly dependent.
\end{proposition}

\begin{proof}
\textbf{Step 1.} Consider a continual-learning-based theory of consciousness, wherein $\mathrm{\textit{pred}}$ is grounded in the ongoing process of plasticity, while $\mathrm{\textit{inf}}$ is derived from behavioral outputs that evince observable evidence of learning.

\textbf{Step 2.} Consider latent learning: a system can undergo plastic changes that do not immediately manifest in behavior. Let $o$ and $o'$ be observations from two systems with identical behavioral outputs (a static substitution) but different plasticity states.

\textbf{Step 3.} Since behavioral outputs are identical, $\mathrm{\textit{inf}}(o) = \mathrm{\textit{inf}}(o')$. However, since plasticity states differ, $\mathrm{\textit{pred}}(o) \neq \mathrm{\textit{pred}}(o')$.

\textbf{Step 4.} If $\mathrm{\textit{pred}}(o) \neq \mathrm{\textit{pred}}(o')$ following a static substitution then strict dependency cannot hold for the theory.
\end{proof}

Finally, it is worth noting that there exists a relationship here to the classic philosophical issues of phenomenological consciousness vs. access consciousness \cite{block2011perceptual}. Learning is about counterfactuals, generalization, and dispositions, and so a theory's predictions grounded in learning (say, focused on physical plasticity) ``overflow'' the access of individual experimental inferences, since the subject of those predictions (like plastic changes) implicitly have effects on a much larger domain (the entire set of all possible experiments, not just the \textit{current} experiment).

\section{Discussion}

\subsection{A Progressive Research Program on Consciousness}

The literature around formal requirements for theories of consciousness, which so far has focused on falsification \cite{doerig2019unfolding} and non-triviality \cite{kleiner2021falsification} and also ``toy'' models \cite{albantakis2025utility}, may to outsiders at first seem a small subfield of the neuroscience of consciousness focused on highly-theoretical ``armchair'' arguments. For this reason, some in the literature have argued that formal or theory-based arguments should take a backseat to empirical research, and we should not put strong constraints on theories of consciousness. \cite{usher2023philosophical}

The results herein show the opposite. Formal requirements for theories of consciousness to be testable offer the glimpse of a rare \textit{progressive} subfield of consciousness research. To put it bluntly: the past two decades have seen a large number of novel theories of consciousness advanced, and empirical research has, so far, not ruled out a single theory of consciousness that I'm aware of. Nor has it answered any urgent questions we desire of consciousness research, like the issue of LLM consciousness.

The point of this paper goes beyond its contents: the goal is to illustrate how a progressive approach to a science of consciousness might take place via painstakingly working out all the requirements theories of consciousness need to satisfy in order not to be a priori falsified and yet still remain non-trivial and scientifically informative. That this already rules out a great number of contemporary theories, gives us the structure of a disproof of LLM consciousness, and points us toward the capability of continual learning, is the kind of actually surprising information that indicates a progressive research program. 

Indeed, this allows us to define a precise target for theoretical consciousness research: a Minimal Non-trivial and Falsifiable Theory of Consciousness (MNFTC). An MNFTC is a theory of consciousness that, via satisfying continual lenient dependency, does not succumb to any of the logical problems around falsifiability and testability that the supermajority of contemporary theories fall into. Any universal substitutions are more like simulations of the properties that a theory's predictions are sensitive to. There is a large amount of future research to be done sketching out more formally what this progressive research program would consist of, delineating its boundaries, and also evaluating potential theories that actually fit those constraints, and connecting this to other work \cite{kleiner2025logic}. And while this paper focuses on learning as one option, it's worth noting there may be other cases of lenient dependency beyond continual learning that are more exotic (speculative alternatives include, say, observer effects in quantum mechanics being somehow immune to substitutions \cite{hameroff2014consciousness}). 

\subsection{Continual Learning as Central to Consciousness}

It is not a coincidence that LLMs, despite their intelligence, still feel like ``ghosts'' \cite{karpathy2025ghosts}. If consciousness is truly based on continual learning (or requires it for some further property), then it makes sense that the substitution of a frozen I/O echo of humans would be technically impressive but fail as our substitutions in ways that are hard to specify or quantify. Often science proceeds spurred on by metaphors from technology: LLMs being close substitutions to humans for many I/O text-based or screen-based tasks, but then also being somehow incomplete substitutions in ways difficult to quantify or specify, seems to preface the results herein about how I/O equivalence is a less valid target than first appears.

Currently, even the best contemporary LLMs struggle on long-term context and agential tasks, like playing Pokémon without a hand-holding scaffold \cite{angliss2025benchmark}. Despite all their impressive gains on benchmarks and in real-world usage, they are still extremely inefficient in their learning and this does have significant functional effects \cite{bell2025future}. This has led to many speculating about the necessity of some unknown type of continual learning as being necessary for true artificial general intelligence \cite{yildiz2024investigating}, as it is indisputable that even a human baby is orders of magnitude more data-efficient than an LLM \cite{warstadt2023findings, frank2023bridging}. If LLMs are indeed non-conscious, this too should tilt us toward theories of consciousness based on learning, since even far less intelligent animals are much more data-efficient in their learning than contemporary AI.

Of course, it is worth noting that there are some theories of consciousness that are already focused on learning and plasticity \cite{cleeremans2011radical, o2025delay}, including (depending on variant and interpretation) some popular theories based on minimizing prediction error and predictive processing that involve updating representations like world models \cite{friston2010free, seth2013interoceptive, clark2013whatever, safron2022integrated}. The results herein should make these prime candidates for exploration and testing. It is a matter of future research to what degree existing theories explicitly fit the requirements of continual lenient dependency, as some may be based only on past learning, rather than current learning.

Intriguingly, the disproof herein does not rule out LLM consciousness during their training in the same way it does for the deployed ``dead'' artificial neural network; however, nor does it guarantee their consciousness in any way (i.e., training might not be similar to the actual kind of continual learning involved in consciousness, nor be technically continual in practice), and so this remains an area for future research. Even if continual learning (which satisfies Proposition 5.1, the Continual Requirement for Lenient Dependency) were to be introduced in LLMs (see attempts such as, e.g., \cite{behrouz2024titans}), the results herein should give us pause. They indicate that non-conscious systems can produce highly elaborate claims to consciousness and this should considerably raise the bar of suspicion. Additionally, continual learning here is given as merely a necessary condition, not a sufficient one, and the results also indicate that LLM consciousness may be limited by degree of plasticity and how that plasticity is distributed (e.g., a small module that continually learns appended to an LLM would not necessarily render the LLM conscious of its conversations). At least when it comes to training, it should be noted that there are detectable differences between how humans vs. LLMs learn, already indicating that the learning process itself may differ significantly and so indicating that LLM consciousness may not hold during training \cite{pesnot2025shared}.

Finally, with regard to humans, a view of the biological brain as a gestalt being continuously updated by literally every experience is, at least, scientifically evocative, and not contradicted by the neuroscience of learning. E.g., place fields can develop in a single trial \cite{bittner2017behavioral}, and significant plasticity can occur over seconds \cite{chaloner2022multiple}; even the Overfitted Brain Hypothesis hypothesizes that dreams may have evolved to combat daily overfitting \cite{hoel2021overfitted}. It may indeed be that every experience must leave a trace for it to be an experience at all.

Consider the difference between how humans and LLMs conduct a conversation: for us humans, each conversation necessarily relies on continual learning to provide context about what has just been said. For the LLM, the whole conversation is stored verbatim externally and is continually being fed in the format of [input $(x, \text{history})$] as part of each new prompt itself. Humans, meanwhile, don't need to be prompted with the entire conversation just to add the next remark.

Therefore, it may well be that, from the beginning, evolution has been focused on continual learning in an organism's appropriate domain (its umwelt), rather than intelligence. This would explain our current uncertainty when high intelligence---at least on things like tests and benchmarks---does not automatically lead to consciousness.

\section{Acknowledgments}

Thank you to the supporters of Bicameral Labs; thanks as well to Johannes Kleiner and Abel Jansma for conversations and feedback. 

\bibliographystyle{unsrt}  %
\bibliography{references}  %

\end{document}